%% file: main.tex
\documentclass[letterpaper, 10 pt, conference]{ieeeconf}
\IEEEoverridecommandlockouts   
\usepackage{max_paper}

\title{\LARGE \bf
Constructive Safety-Critical Control for a Class of Underactuated Systems: A Hierarchical Approach}

\author{Massimiliano de Sa and Aaron D. Ames%
\thanks{This research is supported by AFOSR Grant No. 113535-19668, Hybrid Dynamics - Deconstruction and Aggregation (HyDDRA).}%
\thanks{The authors are with the Department of Control and Dynamical Systems, California Institute of Technology, Pasadena CA 91125, U.S.A.
        {\tt\small \{mdesa, ames\}@caltech.edu}}%
}

\begin{document}

\maketitle
\thispagestyle{empty}
\pagestyle{empty}
\renewcommand{\bf}{\mathbf{f}}

\begin{abstract}

Underactuated systems with nontrivial geometry are abundant in practical robotic problems. However, little is known about constructive safety-critical control design in the presence of underactuation. In this work, we provide a constructive framework for synthesizing safe control architectures and control barrier functions for a class of quadrotor-like underactuated systems. By studying the geometry of their underactuation, we characterize when these systems have natural layered architectures useful for safety-critical control design. Using these architectures, we devise a constructive, hierarchical CBF synthesis procedure and provide conditions under which it succeeds.

\end{abstract}

\input{intro}
\input{prelim}
\input{actuation}
\input{cbf-synth}

\input{conclusion}


\bibstyle{ieeetr}
\bibliography{references}

\end{document}

%% file: intro.tex
\section{Introduction}
Underactuated mechanical systems---systems with fewer degrees of actuation than degrees of freedom---comprise an important class of control systems ubiquitous in modern robotic problems. Due to their underspecified nature, underactuated systems are substantially more challenging to control than their fully-actuated counterparts. Furthermore, numerous important underactuated control systems---for instance, the quadrotor---have manifold-valued states, compelling controller designs to contend with both underactuation \textit{and} nontrivial state space geometry.

In practice, such systems are subject to a host of safety constraints; quadrotors, for instance, must avoid collisions with their environment. \textit{Control barrier functions (CBFs)} have been established as a leading framework for the specification of safety constraints and design of safety-critical controllers \cite{ames2016control, ames2019control}. To leverage this framework, and the host of safety-critical control design techniques it enables, one must first construct a control barrier function. Though the process of synthesizing CBFs for fully-actuated systems is relatively well understood, less is known about the underactuated case.

Recently, a \textit{hierarchical} perspective on synthesizing CBFs for complex, potentially underactuated systems has emerged \cite{molnar2021model, cohen2024safety, cohen2025safety, yamaguchi2026layered, matni2024quantitative}. Within this paradigm, one first designs a CBF for a reduced-order model, a \textit{virtual} model for which CBF design is simple. Then, one assumes there exists a controller for the full-order system that tracks the inputs to the reduced-order model, as well as a Lyapunov certificate of tracking. Under various assumptions (e.g. sufficiently fast tracking, bounded full-to-reduced-order model error), the Lyapunov certificate can be combined with a CBF for the reduced-order model to produce a CBF for the full system.

\begin{figure}
    \centering
    \includegraphics[width=0.97\linewidth]{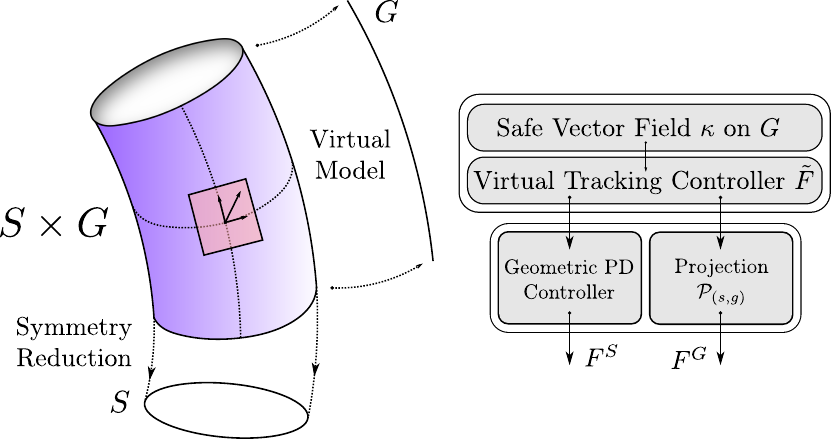}
    \caption{For a class of quadrotor-like underactuated systems, symmetries enable natural layered control architectures that can be used to synthesize closed-form CBFs.}
    \label{fig:reduction}
    \vspace{-10pt}
\end{figure}

Although these hierarchical techniques have demonstrated considerable success as a framework, little is known about hierarchies from the constructive perspective. Indeed, given a complex, underactuated system, it is unknown how to choose a reduced-order model satisfying the requisite tracking and boundedness conditions (nor whether one exists), nor how to find the tracking Lyapunov certificate required for CBF construction (nor whether one exists). Further, the nontrivial role of geometry---which poses fundamental obstructions to tracking control \cite{koditschek1989application, bullo1999tracking}---has yet to be examined.

To better understand when these unknowns have concrete solutions, we perform a geometric study of layered control architectures. We are driven by two principal directions. On one front, we aim to determine how and when a layered architecture can be designed. Here, we ask what conditions enable the existence of ``natural" layered architectures, and how inputs to reduced-order models can be transformed into references for the full-order system to track. On another front, given an architecture in which a safe controller for a reduced-order model is tracked by a controller for the full system, we ask how and when safety certificates for the full system can be assembled in closed form.

In pursuing these two directions, we focus on a class of quadrotor-like geometric mechanical control systems, for which concrete answers are available and interpretable. By leveraging symmetries of these systems, and drawing analogies to existing quadrotor control architectures \cite{lee2010control}, we characterize when natural choices of architectures exist for these systems and when they can be certified by constructive, closed-form control barrier functions.

Our main contributions are summarized as follows:
\begin{enumerate}
    \item Develop and characterize \textit{shape maps} as a means of dealing with underactuation and transforming virtual control inputs into references for underactuated systems.
    \item Use symmetries and shape maps to design layered control architectures for a class of underactuated systems.
    \item Apply these layered architectures to devise a constructive, hierarchical CBF synthesis procedure for a class of underactuated systems.
\end{enumerate}
The rest of the paper is organized as follows. In Section II, we review the necessary preliminaries. In Section III, we discuss architectures, shape maps, and their existence. In Section IV, we discuss CBF synthesis, and in Section V, we provide concluding remarks.

%% file: prelim.tex
\section{Preliminaries}
\subsection{Differential Geometry}
First, we fix basic geometric notation. All of our definitions are standard, following \cite{lee2012smooth, lee2018introduction}, and we shall develop concepts in further detail as the need arises. Throughout, we assume all objects to be \textit{smooth} ($C^\infty$) for brevity.

Let $Q$ be a smooth $n$-manifold. Given a tangent vector $v_q \in T_q Q$ and a cotangent vector $\omega_q \in T_q^*Q$ (the dual space to $T_qQ$), we write $\braket{\omega_q; v_q}$ for the action of $\omega_q$ on $v_q$. We write $C^\infty(Q)$ for the set of smooth functions $h: Q \to \R$, and $dh_q \in T_q^*Q$ for the differential of $h$ at $q$. For $\psi \in C^\infty(Q \times Q)$, we write $d_1 \psi$ and $d_2 \psi$ for differentials in the first and second arguments. For a smooth map $F: Q_1 \to Q_2$, we write $dF_q$ for its differential (analogous to the Jacobian) and $dF_q^*$ for its adjoint. We write $\X(Q)$ for the set of smooth vector fields on $Q$. For $X \in \X(Q)$, $X_q$ is its value at $q$. A rank-$m$ codistribution $\F$ on $Q$ is a smooth assignment of each $q \in Q$ to an $m$-dimensional subspace $\F_q$ of $T_q^*Q$.

We write $\braket{\cdot, \cdot}$ for a Riemannian metric on $Q$ and $\norm{\cdot}_Q$ for its induced norms on $T_qQ$ and $T_q^*Q$. A metric induces isomorphisms $\sharp$ and $\flat$ taking covectors to vectors and vice versa. The gradient of $h \in C^\infty(Q)$ at $q$ is $\grad h|_q = (dh_q)^\sharp$.

A Lie group $G$ is a group with a smooth manifold structure. We write $L_h: G \to G$ for left translation on the Lie group, $L_h(g) = h \cdot g$, and $e \in G$ for its identity element. The Lie algebra $\g$ of $G$ is its tangent space at the identity, $\g = T_e G$. The dual space to $\g$ is denoted $\g^*$.

\subsection{Simple Mechanical Control Systems}
Lagrangian control systems on $\R^n$ are generalized to a geometric setting within the following framework \cite{bullo2005geometric}.
\begin{defn}\label{defn:smcs}
    A \textit{simple mechanical control system (SMCS)} is a tuple $\Sigma = (Q, \braket{\cdot, \cdot}, V, \F)$ of an $n$-manifold $Q$ (the \textit{configuration manifold}), a Riemannian metric $\braket{\cdot, \cdot}$ on $Q$ (the \textit{kinetic energy metric}), $V \in C^\infty(Q)$ (the \textit{potential function}), and a rank-$m$ codistribution $\F$ (the \textit{input codistribution}).
\end{defn}
\begin{defn}\label{defn:controller}
    A \textit{controller} for a SMCS $\Sigma$ is a map $F: TQ \to \F$ with $F(v_q) \in \F_{q}$ for all $v_q \in TQ$.
\end{defn}
The equations of motion of a SMCS with controller $F$ are:
\begin{align}
    \nabla_{\dot q} \dot q &= - \grad V|_q + F(\dot q)^\sharp. \label{eq:smcs-eom}
\end{align}
Here, $\nabla$ is the \textit{Levi-Civita connection} induced by the metric, an object that enables the computation of directional derivatives, termed \textit{covariant derivatives}, on $Q$. In \eqref{eq:smcs-eom}, $\nabla_{\dot q} \dot q$ is interpreted as the acceleration of $q$. Provided $F$ is locally Lipschitz, \eqref{eq:smcs-eom} has a unique maximal trajectory $q: I(v_q) \to Q$ for each $v_q \in TQ$, a curve satisfying $\dot q(0) = v_q$ and \eqref{eq:smcs-eom} for all times in its maximal interval of existence $I(v_q) \subseteq \R$.

\subsection{Error Functions \& Transport Maps}
To design trajectory tracking controllers for a SMCS, one requires notions of \textit{errors} in configuration and velocity. Configuration error is encoded by the following \cite{bullo2005geometric}.
\begin{defn}\label{defn:config-error-func}
    Let $Q$ be a Riemannian manifold. A smooth map $\psi: Q \times Q \to \R_{\geq 0}$ is a \textit{configuration error function} if:
    \begin{enumerate}
        \item $\psi(q, r) = \psi(r, q) \; \forall q, r \in Q$ and $\psi(q, r) = 0 \Leftrightarrow q = r$.
        \item $\psi_r: q \mapsto \psi(q, r)$ is proper\footnote{A function is \textit{proper} if its preimage of every compact set is compact.} for all $r \in Q$.
        \item $d_1 \psi_{(r, r)} = 0_r$ for all $r \in Q$.
        \item $\Hess_1 \psi|_{(r, r)}$ is positive definite for all $r \in Q$, for $\Hess_1 \psi$ the Hessian\footnote{The Hessian of $h \in C^\infty(Q)$ at $q$ is the symmetric, bilinear form $\Hess h|_q : T_q Q \times T_q Q \to \R$, $\Hess h|_q(v_q, w_q) = \braket{v_q, \nabla_{w_q}\grad h|_q}$.} of $\psi$ in its first argument.
    \end{enumerate}
    Further, $\psi$ is $L$-\textit{uniformly quadratic} with $L > 0$ if $\exists \; b_1, b_2 > 0$ such that for all $r \in Q$ and $q \in \psi_r^{-1}(\R_{\leq L})$,
    \begin{align}
        b_1 \norm{d_1 \psi_{(q, r)}}_Q^2 \leq \psi(q, r) \leq b_2\norm{d_1 \psi_{(q, r)}}_Q^2. \label{eq:unif-quadr}
    \end{align}
\end{defn}
Velocity error is defined through a \textit{transport map}, which enables the comparison of vectors in different tangent spaces.
\begin{defn}
    A \textit{transport map} $\T$ is a collection of maps $\T_{(q, r)}$ satisfying the following properties:
    \begin{enumerate}
        \item For all $q, r \in Q$, $\T_{(q, r)}: T_r Q \to T_q Q$ is a linear map.
        \item For all $r \in Q$, $\T_{(r, r)} = \id_{T_rQ}$.
        \item For each $w_r \in TQ$, $Y: q \mapsto \T_{(q, r)}w_r$ belongs to $\X(Q)$. 
    \end{enumerate}
    Given a configuration error function $\psi$ on $Q$, $\T$ is \textit{compatible} with $\psi$ if $d_2 \psi_{(q, r)} = -\T_{(q, r)}^*(d_1 \psi_{(q, r)})$ for all $q, r \in Q$.
\end{defn}
We define the \textit{covariant derivative of $\T$} as $\nabla \T_{(q, r)}: T_q Q \times T_r Q \to T_q Q$, $\nabla \T_{(q, r)}(v_q, w_r) = \nabla_{v_q} (\T_{(q, r)}w_r)$.
\begin{example}\label{ex:so3-config-error}
    On $SO(3)$, the space of 3D rotation matrices, $\psi(R_1, R_2) = \frac{1}{2}\tr(I - R_2^\top R_1)$ is a configuration error function. Writing tangent vectors in $T_RSO(3)$ as $R \hat \Omega$, for $\wedge: \R^3 \to \so(3)$ the hat map, a compatible transport map is given by $\T_{(R_1, R_2)}(R_2 \hat \Omega_2) = R_2 \hat \Omega_2 R_2^\top R_1$  \cite[p. 554]{bullo2005geometric}.
\end{example}
\subsection{Control Barrier Functions}
In the \textit{control barrier function (CBF)} framework, one aims to design a controller to keep a system within a set $\C$ of safe states. Such control design is achievable when $h$ is a CBF; we now recall the definition of a CBF for a SMCS \cite{de2026bundles}.
\begin{defn}
    Let $\Sigma$ be a SMCS and $D$ be an open subset of $TQ$. A function $h \in C^\infty(D)$ is a \textit{control barrier function (CBF)} for $\Sigma$ on $D$ with safe set $\C = h^{-1}(\R_{\geq 0})$ if there is a function\footnote{A function $\alpha: \R \to \R$ belongs to $\K_\infty^e$ if it is continuous, strictly increasing, and satisfies $\alpha(0) = 0$, $\lim_{r \to \pm \infty} \alpha(r) = \pm \infty$.} $\alpha \in \K_\infty^e$ for which
    \begin{align}
        \sup_{F_q \in \F_q} \dot h(v_q, F_q) > -\alpha(h(v_q)), \quad \forall v_q \in \C, \label{eq:cbf-defn}
    \end{align}
    where $\dot h$ is the Lie derivative of $h$ at $(v_q, F_q)$ subject to \eqref{eq:smcs-eom}.
\end{defn}
Given a CBF $h$ for a SMCS $\Sigma$ and a nominal, potentially unsafe controller $ F_d : TQ \to \F$, the CBF-QP \textit{safety filter} is
\begin{align}
    F_{\text{QP}}(v_q) =& \argmin_{F_q \in \F_q} \lVert{F_q - F_d(v_q)}\rVert_Q^2 \label{eq:cbf-qp} \\
     \text{ s.t. }& \dot h(v_q, F_q) \geq -\alpha(h(v_q)). \nonumber
\end{align}
For locally Lipschitz $F_d$, $\alpha$, this produces a locally Lipschitz control law on $\C$ that minimally modifies a nominal input $F_d$ to a safe input $F_{\text{QP}}$. For any initial $v_q \in \C$, the resulting trajectory of $\Sigma$ under $F_{\text{QP}}$ remains in $\C$ for all positive times in its maximal interval of existence, thereby achieving safety. Further, it is possible to construct smooth approximations of \eqref{eq:cbf-qp}, termed \textit{smooth safety filters}, in closed-form \cite{cohen2023characterizing}.

%% file: actuation.tex
\section{Architectures \& Actuation Geometry}
In the literature, control barrier functions for complex, potentially underactuated systems are often constructed \textit{hierarchically} from reduced-order models. For mechanical systems, for instance, hierarchical CBF synthesis typically proceeds as follows. Given a SMCS on a configuration manifold $Q$, a set of safe configurations $\C_0 \subset G$ is specified, where $G$ is a smaller manifold related to $Q$ by a projection $\pi: Q \to G$. One designs a CBF for $\Sigma$ with safe set $\C$ such that $v_q \in \C$ implies $q$ is safe with respect to $\C_0$, satisfying $\pi(q) \in \C_0$.

The design procedure is as follows. On $G$, choose a kinematic reduced-order model, a vector field $\kappa$ that is safe with respect to $\C_0$. Assuming that the full mechanical system can track $\kappa$ in a certifiable manner, the certificate of tracking can be combined with a certificate of safety for $\kappa$ to produce the desired CBF for the full system \cite{molnar2021model, cohen2025safety, de2026bundles}.

Though this hierarchical process---determining a safe vector field, tracking it, and finding certificates of tracking and safety---is well understood for \textit{fully-actuated} systems, comparatively little is known about the underactuated case. This leads to the following motivating problem of this work.
\begin{prob}\label{prob:cbf-constr}
    Let $\Sigma$ be an underactuated SMCS on $Q$ and $\pi: Q \to G$ be a projection to a manifold $G$. Given a safe configuration set $\C_0 \subset G$, and a safe vector field $\kappa$ on $G$, 
    \begin{enumerate}
        \item How and when can a $\kappa$-tracking controller be designed and certified by a closed-form tracking certificate?
        \item When can the tracking certificate be used to synthesize a CBF for the safe configuration set $\tilde \C_0 = \pi^{-1}(\C_0)$?
    \end{enumerate}
\end{prob}
To gain insight into this problem, we begin by examining one of the most important underactuated control systems: the quadrotor on $SE(3)$. Given a set of safe positions $\C_0 \subset \R^3$, we begin by choosing a safe vector field $\kappa$ on $\R^3$ for $\C_0$. To track $\kappa$, a popular strategy is to use yet another hierarchical control design, as visualized in Fig. \ref{fig:qrotor} \cite{lee2010control}. First, one designs a force controller for a \textit{virtual} system, a fully-actuated point mass model on $\R^3$; for each $v_x \in T\R^3$, this produces a desired tracking force $f_d(v_x)$. Recognizing that the quadrotor can only exert $f_d(v_x)$ if it aligns $R e_3$ with $f_d(v_x)$, one computes a desired orientation $R_d \in SO(3)$, with $R_d e_3$ aligned with $f_d(v_x)$. The final control law is computed by projecting $f_d(v_x)$ onto $Re_3$ and tracking the desired orientation $R_d$. Provided this tracking is sufficiently fast, this design enables stable tracking of $\kappa$ with a closed-form certificate of tracking \cite{lee2010control}.
\begin{figure}
    \centering
    \includegraphics[width=0.8\linewidth]{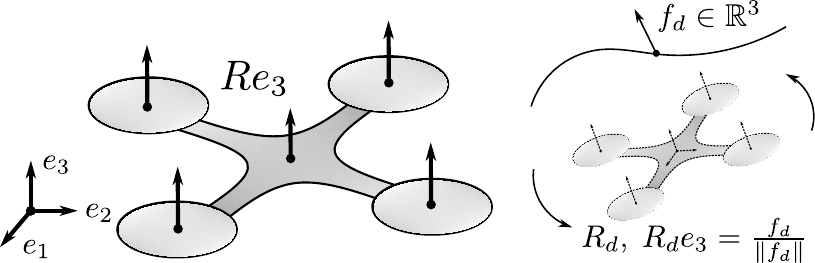}
    \caption{A quadrotor can only apply thrust in the $Re_3$-direction; to track $f_d \in \R^3$, one must \textit{reorient} the quadrotor to align $Re_3$ with $f_d$.}
    \label{fig:qrotor}
    \vspace{-10pt}
\end{figure}

In this section, starting from the key lessons taught by the quadrotor architecture, we address the inverse problem of identifying the key structural aspects which make the quadrotor $\kappa$-tracking architecture possible.

\subsection{Trivial Simple Mechanical Control Systems}
To begin, we seek a class of systems that admit simple dynamical decompositions and virtual models, analogous to the quadrotor. The existence of decompositions of the quadrotor can be traced to an important property: \textit{symmetry}. Consider the Lie group $(\R^3, +)$ and the Lie group action\footnote{A \textit{Lie group action} is a smooth family of transformations parameterized by a Lie group; we refer to \cite[Ch. 5]{bullo2005geometric} for further details.} $\Phi_h: SE(3) \to SE(3)$ of \textit{left translation} by $h \in \R^3$, $\Phi_h((R, x)) = (R, L_h(x)) = (R, h + x)$. For each $h \in \R^3$, both the kinetic energy metric and input codistribution of the quadrotor are \textit{invariant} under the translation $(R, x) \mapsto (R, h + x)$; thus, both have a \textit{translational symmetry}.

This setting of a SMCS on a configuration space $Q = S \times G$, for $S$ a smooth manifold and $G$ a Lie group, with objects invariant under the group action $\Phi_h(s, g) = (s, L_h(g))$, leads to a general class of quadrotor-like underactuated systems.  Below, we employ the standard identifications $T_{(s, g)}(S \times G) \cong T_s S \oplus T_gG$ and $T_{(s, g)}^*(S \times G) \cong T_s^*S \oplus T_g^*G$.
\begin{defn}\label{defn:triv-smcs}
    A \textit{trivial SMCS} $\Sigma = (Q, \braket{\cdot, \cdot}, V, \F)$ satisfies:
    \begin{enumerate}
        \item $Q = S \times G$, for $S$ a manifold called the \textit{shape space} and $G$ a Lie group called the \textit{symmetry group}.
        \item $\braket{(v_s, v_g), (w_s, w_g)} = \braket{v_s, w_s}_S + \braket{v_g, w_g}_G$, for $\braket{\cdot, \cdot}_S$ a metric on $S$ and $\braket{\cdot, \cdot}_G$ a left-invariant metric on $G$.
        \item $V(s, g) = V_{G}(g)$ for some $V_G \in C^\infty(G)$.
        \item $\F$ is invariant under $\Phi_h: (s, g) \mapsto (s, L_h(g))$:
        \begin{align*}
            d(\Phi_h)_{(s, g)}^* \F_{(s, hg)} = \F_{(s, g)}, 
        \end{align*}
        and $T_s^*S \oplus \{0_g\} \subseteq \F_{(s, g)}$, for all $s \in S$ and $g, h \in G$.
    \end{enumerate}
\end{defn}
Thus, as is the case for the quadrotor, the underactuation of a trivial SMCS is concentrated in $G$, while $S$ is fully-actuated. The equations of motion of a trivial SMCS neatly decompose into a $G$-subsystem $\Sigma_G$ and an $S$-subsystem $\Sigma_S$:
\begin{align*}
    \Sigma_G: \nabla_{\dot g}^G \dot g &= - \grad V_G|_g + (F_g^G)^\sharp \\
    \Sigma_S: \nabla_{\dot s}^S \dot s &= (F_s^S)^\sharp,
\end{align*}
where $(F_s^S, F_g^G) \in \F_{(s, g)}$ and $\nabla^G, \nabla^S$ are the Levi-Civita connections on $G$ and $S$. Critically, the only coupling between $\Sigma_G$ and $\Sigma_S$ is through the structure of the force codistribution $\F$. Though this is certainly \textit{not} the case for general underactuated systems, we will find that this structure enables precise insights into underactuated control. With this in mind, we define a \textit{virtual model}. 
\begin{defn}
    The \textit{virtual model} of a trivial SMCS $\Sigma$ is the SMCS $\tilde \Sigma_G = (G, \braket{\cdot, \cdot}_G, V_G, T^*G)$.
\end{defn}
\begin{example}\label{ex:qrotor-triv-smcs}
    The quadrotor on $SE(3)$ is a trivial SMCS with:
    \begin{enumerate}
        \item \textit{Configuration manifold}: $S = SO(3)$ and $ G = \R^3$.
        \item \textit{Metric}: for $m$ the quadrotor's mass and $J \in \R^{3\times 3}$ its inertia\footnote{We identify tangent vectors in $T_x \R^3$ and $T_R SO(3)$ with vectors in $\R^3$ and matrices $R \hat \Omega, \; \Omega \in \R^3$, respectively.}, $\braket{(R \hat \Omega_1, v), (R \hat \Omega_2, w)} = \Omega_1^\top J \Omega_2 + mv^\top w$.
        \item \textit{Potential}: for $x \in \R^3$, $V_G(x) = mg e_3^\top x$.
        \item \textit{Force codistribution}: for $(R, x) \in SE(3)$ and $e_3 = (0, 0, 1)$, $\F_{(R, x)} = T_R^* SO(3) \oplus \spn(Re_3)^\flat$.
    \end{enumerate}
    The virtual model of this system is $(\R^3, \braket{\cdot, \cdot}_G, V_G, T^* \R^3)$, a fully-actuated point mass with gravity.
\end{example}

\subsection{Shape Maps}
Trivial SMCS contribute the first key component of the quadrotor's hierarchical control design: a \textit{virtual} model. The second component enabling this design was a map from virtual forces $f_d$ to orientations in $SO(3)$ at which $f_d$ belongs to the force codistribution. In the trivial SMCS setting, this amounts to finding a map from a force $F_g \in T^* G$ to a shape $s \in S$ at which $F_g$ belongs to the input codistribution. We term this object a \textit{shape map}. To give a formal definition, we require the following additional notation. Let $\F_{(s, g)}^G$ be the $G$-component of $\F_{(s, g)}$, the unique subspace of $T_g^*G$ for which $\F_{(s, g)} = T_s^*S \oplus \F_{{(s, g)}}^G$, $\Pc_{(s, g)}$ be the orthogonal projection from $T_g^*G$ to $\F^G_{(s, g)}$, and $Z^*(G)$ the set of all zero covectors in $T^*G$.
\begin{defn}\label{defn:shape-map}
    A \textit{local shape map} for a trivial SMCS $\Sigma$ is a smooth map $\Sc: U \to S$, for $U \subseteq T^*G \setminus Z^*(G)$ open, with:
    \begin{enumerate}
        \item \textit{Actuation}: for all $g \in G$ and $F_g \in U$, $F_g \in \F^G_{(\Sc(F_g), g)}$.
        \item \textit{Scale-invariance}: for all $\lambda > 0$, $g \in G$, and $F_g \in U$,
        \begin{align*}
            \lambda F_g \in U \text{ and } \Sc(F_g) = \Sc(\lambda F_g).
        \end{align*}
        \item \textit{Translation-invariance}: for all $h, g \in G$ and $F_{hg} \in U$,
        \begin{align*}
            d(L_h)_{g}^* F_{hg} \in U \text{ and } \Sc(F_{hg}) = \Sc( d(L_h)_{g}^* F_{hg}).
        \end{align*}
    \end{enumerate}
    If $U = T^*G \setminus Z^*(G)$, we refer to $\Sc$ as a \textit{shape map.} If $\psi$ is a configuration error function on $S$, then $\Sc$ is \textit{$L$-compatible}\footnote{As $\psi$ is locally quadratic, $L$-compatibility is locally a mild condition.} with $\psi$ if $\exists \; L, C > 0$ for which $\psi(s, \Sc(F_g)) \leq L$ implies:
    \begin{align*}
        \norm{\Pc_{(s, g)}\Big(\tfrac{F_g}{\norm{F_g}_G}\Big) - \tfrac{F_g}{\norm{F_g}_G}}_G^2 &\leq C \psi(s, \Sc(F_g)).
    \end{align*}
\end{defn}
Thus, a shape map simply takes in a force direction, specified by a nonzero covector, and returns a shape at which it belongs to the force codistribution. $L$-compatibility, which controls force projection error in terms of a configuration error, will later prove valuable in control design.

Given the role that shape maps play in the quadrotor's hierarchical design, we now seek conditions for their existence. First, using the symmetry of $\F$, we reduce the search for shape maps to a search over a simpler class of functions.
\begin{lem}\label{lem:reduce-struct}
    Let $\Sigma$ be a trivial SMCS, $\af_s = \F_{(s, e)}^G \subseteq \g^*$, and $\sph^{\g^*} = \{ \omega \in \g^* : \norm{\omega}_G = 1\}$ the unit sphere in $\g^*$. Then:    \begin{enumerate}
        \item For each $(s, g) \in S \times G$, $\F_{(s, g)}^G = d(L_{g^{-1}})_{g}^*\af_s$.
        \item Let $\mathfrak U$ be an open subset of $\sph^{\g^*}$. For any smooth $\s: \mathfrak U \subseteq \sph^{\g^*} \to S$ satisfying $f \in \af_{\s(f)}$ for all $f \in \mathfrak U$, the map
        \begin{align*}
            \Sc : F_g \mapsto \s\Big(d(L_g)_{e}^*\tfrac{F_g}{\norm{F_g}_G}\Big)
        \end{align*}
        is a local shape map on $U = \bigcup_{g \in G, \lambda > 0} d(L_{g^{-1}})_{g}^*(\lambda \mathfrak U)$. Conversely, every local shape map is of this form.
    \end{enumerate}
\end{lem}
\begin{proof}
    Item (1) follows immediately from invariance of $\F$. Next, we prove (2). Suppose $\s : \Uf \to S$ satisfies $f \in \af_{\s(f)}$ for all $f \in \Uf$. First, we note that the domain $U$ is an open subset of $T^*G \setminus Z^*(G)$. Now, we verify that the induced $\Sc$ satisfies the actuation property. Let $F_g \in U$. For $f = d(L_g)_{e}^* \tfrac{F_g}{\norm{F_g}_G}$, we have that $f \in \Uf$ and that $F_g = \norm{F_g}_G d(L_{g^{-1}})_{g}^* f$. Thus, using that $f \in \af_{\s(f)}$ and $\F_{(s, g)}^G = d(L_{g^{-1}})_{g}^*\af_s$, we find:
    \begin{align*}
        F_g &= \norm{F_g}_G d(L_{g^{-1}})_{g}^* f \in d(L_{g^{-1}})_{g}^* \af_{\s(f)} = \F_{(\s(f), g)}^G.
    \end{align*}
    Since $\Sc(F_g) = \s(f)$, we conclude that $F_g \in \F_{(\Sc(F_g), g)}^G$. The scale and translation-invariance properties of $\Sc$ follow by construction of $\Sc$ and $U$. To prove the converse of item (2), given a local shape map $\Sc: U \to S$, we define $\s$ by
    \begin{align*}
        \s(f) = \Sc(f), \; \forall f \in \Uf = U \cap \sph^{\g^*}.
    \end{align*}
    This map $\s$ induces $\Sc$ in the same manner as above.
\end{proof}

We shall call maps of the form of $\s$ \textit{local Lie shape maps}. Lemma \ref{lem:reduce-struct} reduces the question of existence of local shape maps to that of local Lie shape maps. To establish conditions for the existence of these maps, our next step is to assemble the unit vectors of the $\af_s$ subspaces into a smooth manifold.
\begin{lem}\label{lem:A-defn}
    For a trivial SMCS $\Sigma$, $\Ac := \bigcup_{s \in S} \{s\} \times (\af_s \cap \sph^{\g^*})$
    is a smooth, embedded submanifold of $S \times \sph^{\g^*}$.
\end{lem}
\begin{proof}
    Let $r = \dim \af_s^\perp$ (since $\F$ has constant rank, $r$ is independent of $s$), and fix $(s_0, f_0) \in \Ac$. As $\F$ is a smooth codistribution, there are $r$ smooth functions $\theta_i: V \to \g^*$ on a neighborhood $V$ of $s_0$ in $S$, for which $\af_s^\perp = \spn\{\theta_i(s)\}_{i = 1}^r$ on $V$. Define a map $\tilde \Theta: V \times \g^* \to \R^r$ by
    \begin{align*}
        \tilde \Theta: (s, f) \mapsto (\braket{\theta_1(s), f}, ..., \braket{\theta_r(s), f}),
    \end{align*}
    where $\braket{\cdot, \cdot}$ is the inner product on $\g^*$ induced by the metric. To prove that $\Ac$ is an embedded submanifold of $S \times \sph^{\g^*}$, we shall construct a \textit{local defining function} for $\Ac$ from $\tilde \Theta$. To this end, we study the differential of $\tilde \Theta$ at $(s_0, f_0)$,
    \begin{align*}
        d(\tilde \Theta)_{(s_0, f_0)}(v_{s_0}, w_{f_0}) &= d_1 (\tilde \Theta)_{(s_0, f_0)}(v_{s_0})\\
        &+ d_2(\tilde \Theta)_{(s_0, f_0)}(w_{f_0})
    \end{align*}
    Under the canonical identifications $ T_{\tilde \Theta(s_0, f_0)} \R^r \cong \R^r$ and $T_{f_0} \g^* \cong \g^*$, $d_2 \tilde \Theta_{(s_0, f_0)}$ is the linear map $w_{f_0} \mapsto \tilde \Theta(s_0, w_{f_0})$, and is therefore surjective. We argue its restriction to $T_{f_0} \sph^{\g^*}$ remains surjective. As a subspace of $\g^*$, $T_{f_0} \sph^{\g^*}$ is the orthogonal complement of $\spn\{f_0\}$. Since $\theta_i(s_0) \perp f_0$ for all $i$, $(T_{f_0}\sph^{\g^*})^\perp \subseteq \ker d_2(\tilde \Theta)_{(s_0, f_0)}$. Decomposing $T_{f_0} \g^* = T_{f_0} \sph^{\g^*} \oplus (T_{f_0}\sph^{\g^*})^\perp$, we then conclude that surjectivity of $d_2(\tilde \Theta)_{(s_0, f_0)}$ implies surjectivity of its restriction to $T_{f_0} \sph^{\g^*}$. Thus, for $\Theta = \tilde \Theta|_{V \times \sph^{\g^*}}$, $d_2 \Theta_{(s_0, f_0)}$ is surjective. In turn, $d \Theta_{(s_0, f_0)}$ is surjective in a neighborhood $W$ of $(s_0, f_0)$ in $S \times \sph^{\g^*}$. Since $\Ac \cap W = \Theta|_{W}^{-1}(0)$, $\Theta|_W$ is then a local defining function for $\Ac$. Thus, each point of $\Ac$ belongs to the domain of a local defining function, and $\Ac$ is an embedded submanifold of $S \times \sph^{\g^*}$ \cite[Prop. 5.16]{lee2012smooth}.
\end{proof}
As $\Ac$ consists of all $(s, f)$ for which $f \in \sph^{\g^*}$ belongs to $\af_s$, we can factor any local Lie shape map $\s$ through $\Ac$ as:
\[\begin{tikzcd}
	{\mathfrak U \subseteq \sph^{\g^*}} & \Ac & S
	\arrow["{f \mapsto (s, f)}", from=1-1, to=1-2]
	\arrow["\s"', curve={height=18pt}, from=1-1, to=1-3]
	\arrow["{(s, f) \mapsto s}", from=1-2, to=1-3].
\end{tikzcd}\]
Thus, each $\s$ is fully determined by a map of the form $f \in \mathfrak U \mapsto (s, f) \in \Ac$. Using this observation, we now precisely characterize the maps $f \mapsto (s, f) \in \Ac$.
\begin{prop}\label{prop:act-cond}
    Let $\Sigma$ be a trivial SMCS with $\dim \af_s^\perp = r$ and $\Theta$ be as in Lemma \ref{lem:A-defn}. Consider condition $(\ast)$:
    \begin{align}
        \parbox{22.3em}{$\forall \; (s_0, f_0) \in \Ac$, there is a neighborhood $V$ of $s_0$ in $S$ and $r$ smooth $\theta_i: V \to \g^*$ spanning $\af_s^\perp$ on $V$ for which $\Theta_{f_0}: s \mapsto (\braket{\theta_i(s), f_0})_{i = 1}^r$ has surjective $d(\Theta_{f_0})_{s_0}$.} \tag{$\ast$} \label{eq:subm-cond}
    \end{align}
    For maps $\pi: \Ac \to \sph^{\g^*}, (s, f) \mapsto f$ and $\pr_S: (s, f) \mapsto s$:
    \begin{enumerate}
        \item $\pi$ is a submersion if and only if \eqref{eq:subm-cond} holds. If $\pi$ is a surjective submersion, then $\pi: \Ac \to \sph^{\g^*}$ is a bundle\footnote{A \textit{bundle} $\pi: E \to Q$ consists of manifolds $E$ and $Q$ and a smooth, surjective submersion $\pi$, a surjective map with surjective differential. A local section of $\pi: E \to Q$ is a map $\sigma: U \subseteq Q \to E$ for which $\pi \circ \sigma = \id_U$.}.
        \item Let $\pi: \Ac \to \sph^{\g^*}$ be a bundle. Every local Lie shape map $\s: \Uf \to S$ determines a unique local section $\sigma: \Uf \to \Ac$ for which $\pr_S \circ \sigma = \s$. Conversely, for every local section $\sigma$, $\pr_S \circ \sigma$ is a local Lie shape map.
    \end{enumerate}
\end{prop}
\begin{proof}
    First, we show (1). Fix $(s_0, f_0) \in \Ac$, and let $\theta_i$ be as in \eqref{eq:subm-cond}. We will show that $d\pi_{(s_0, f_0)}$ is surjective \textit{iff} $d(\Theta_{f_0})_{s_0}$ is. As $d\pi_{(s_0, f_0)}(v_{s_0}, w_{f_0}) = w_{f_0}$, we must prove that
    \begin{align}
        \forall w_{f_0}, \; \exists \; v_{s_0} \text{ s.t. } (v_{s_0}, w_{f_0}) \in T_{(s_0, f_0)}\Ac, \label{eq:first-tan-sp-cond}
    \end{align}
    \textit{iff} $d(\Theta_{f_0})_{s_0}$ is surjective. As shown in the proof of Lemma \ref{lem:A-defn}, the map $\Theta: (s, f) \mapsto (\braket{\theta_i(s), f})$ is a local defining function for $\Ac$ in a neighborhood of $(s_0, f_0)$. Hence, $T_{(s_0, f_0)} \Ac = \ker d\Theta_{(s_0, f_0)}$. Now, recall that $d(\Theta)_{(s_0, f_0)}(v_{s_0}, w_{f_0})$ is computed $d_1(\Theta)_{(s_0, f_0)}v_{s_0} + d_2(\Theta)_{(s_0, f_0)}w_{f_0}$, and, as shown in Lemma \ref{lem:A-defn}, $d_2(\Theta)_{(s_0, f_0)}$ is surjective. Thus, the condition:
    \begin{align*}
        \forall w_{f_0}, \; \exists \; v_{s_0} \text{ s.t. } d(\Theta)_{(s_0, f_0)}(v_{s_0}, w_{f_0}) = 0,
    \end{align*}
    which is equivalent to \eqref{eq:first-tan-sp-cond}, holds if and only if $d_1(\Theta)_{(s_0, f_0)}$ is surjective. As $d_1(\Theta)_{(s_0, f_0)}$ is surjective \textit{iff} $d(\Theta_{f_0})_{s_0}$ is, we conclude $\pi$ is a submersion \textit{iff} \eqref{eq:subm-cond} holds. Since the bundle condition is definitional, item (1) follows. 

    Now, we show (2). Suppose $\s: \Uf \to S$ is a local Lie shape map. Since $f \in \af_{\s(f)}$ for all $f \in \Uf$, $\sigma: f \mapsto (\s(f), f) \in \Ac$ is a local section of $\pi: \Ac \to \sph^{\g^*}$. Further, this is the unique section $\sigma: \Uf \to \Ac$ for which $\pr_S \circ \sigma = \s$. Now, we show the converse. Suppose $\sigma: \Uf \to \Ac$ is a local section. Since for each $f \in \Uf$, $\sigma(f) \in \Ac$ and $\pi(\sigma(f)) = f$, it must be that $f \in \af_{\pr_S \circ \sigma(f)}$. Item (2) follows, completing the proof.
\end{proof}
Together, Lemma \ref{lem:reduce-struct} and Proposition \ref{prop:act-cond} yield the following conditions for the existence of shape maps.
\begin{theo}\label{theo:shape-map}
    Let $\Sigma$ be a trivial SMCS for which \eqref{eq:subm-cond} holds and $\pi$ is surjective. Then, $\Sigma$ has a local shape map in a neighborhood of each $F_g \in T^*G \setminus Z^*(G)$. If $\pi: \Ac \to \sph^{\g^*}$ has a global section, then $\Sigma$ has a global shape map.
\end{theo}
\begin{proof}
    Let $F_g \neq 0$ and $f = d(L_g)_{e}^* \tfrac{F_g}{\norm{F_g}_G}$. Assuming \eqref{eq:subm-cond} and surjectivity, $\pi$ is a surjective submersion, and therefore has a local section on a neighborhood of $f$ \cite[Thm. 4.26]{lee2012smooth}. Applying Proposition \ref{prop:act-cond} and Lemma \ref{lem:reduce-struct}, we obtain a local Lie shape map on a neighborhood of $f$ and a local shape map on a neighborhood of $F_g$. The global case is identical.
\end{proof}

Although there are generally topological obstructions to the existence of global sections---and hence to global shape maps---it is often the case that they exist on large domains.
\begin{example}\label{ex:qrotor-shape}
    We construct a local shape map for the quadrotor on $SE(3)$. First, following \cite{welde2023role}, we consider the map
    \begin{align}
        \sigma&: \sph^2 \setminus \{-e_3\} \to SO(3)\\
        \sigma(p) &= \begin{bmatrix}
            1 - \frac{p_1^2}{p_3 + 1} & \frac{-p_1 p_2}{p_3 + 1} & p_1\\
            \frac{-p_1 p_2}{p_3 + 1} & 1 - \frac{p_2^2}{p_3 + 1} & p_2\\
            -p_1 & -p_2 & p_3
        \end{bmatrix},
    \end{align}
    and define $\Sc$ by $\Sc(f) =\sigma(f/\norm{f}_2)$ for all nonzero $f \neq -\lambda e_3$, $\lambda > 0$. This maps $f$ to $R\in SO(3)$ with $Re_3 = f/\norm{f}_2$. For $\psi$ as in Example \ref{ex:so3-config-error} and $R \in SO(3)$,
    \begin{align*}
         \norm{\Pc_{(R, x)}\Big(\tfrac{f}{\norm{f}_G}\Big) - \tfrac{f}{\norm{f}_G}}_G^2 \leq 2\psi(R, \Sc(f)).
    \end{align*}
\end{example}

%% file: cbf-synth.tex
\section{Hierarchical CBF Synthesis}
Theorem \ref{theo:shape-map} gives conditions establishing the existence of a link---a shape map---between inputs to a virtual system and configurations of the true system. This provides the main ingredient used in the quadrotor architecture of \cite{lee2010control} to manage underactuation. Now, we show that this architecture can be employed more generally for the constructive safety-critical control of trivial simple mechanical control systems.

Below, we describe the problem setting. To maintain a concise exposition, we make several simplifying assumptions.
\begin{setting}\label{setting:lay-arch}
    Let $\Sigma$ be a trivial SMCS and $h_0 \in C^\infty(G)$ have zero a regular value and a compact zero superlevel set $\C_0 = h_0^{-1}(\R_{\geq 0})$ of safe configurations in $G$. Assume:
    \begin{enumerate}
        \item An $L$-uniformly quadratic configuration error function $\psi$ on the shape space $S$, for which 
        \begin{align*}
        \psi \leq L \Rightarrow b_1 \norm{d_1 \psi_{(s, r)}}_S^2 \leq \psi(s, r) \leq b_2 \norm{d_1 \psi_{(s, r)}}_S^2,
        \end{align*}
        and $b_1 \norm{d_1 \psi_{(s, r)}}_S^2 \leq \psi(s, r)$ for all $(s, r) \in S\times S$.
        \item A transport map $\T$ on $S$ that is compatible with $\psi$.
        \item A local shape map $\Sc: U \subset T^* G \setminus Z^*(G) \to S$ that is $L$-compatible with $\psi$, satisfying 
        \begin{align*}
            \norm{\Pc_{(s, g)}\Big(\tfrac{F_g}{\norm{F_g}_G}\Big) - \tfrac{F_g}{\norm{F_g}_G}}_G^2 &\leq C \psi(s, \Sc(F_g)),
        \end{align*}
        for all $s \in S $ and $F_g \in U$ with $\psi(s, \Sc(F_g)) \leq L$. Shrinking $L$ if necessary, assume that $L < \frac{1}{C}$.
    \end{enumerate}
\end{setting}
This section proceeds as follows. First, we will detail a layered architecture for safety-critical control. Then, we will design controllers for each layer of the architecture and produce certificates of their performance. Finally, we will combine these certificates to produce a CBF $h$ for $\Sigma$.

\subsection{Layered Architecture}
Now, recalling the kinematic reduced-order model architecture proposed in Problem \ref{prob:cbf-constr}, and the quadrotor tracking architecture \cite{lee2010control}, we propose the following:
\begin{enumerate}
    \item Design a safe kinematic vector field $\kappa \in \X(G)$ for $\C_0$.
    \item Design a controller $\tilde F$ for the virtual model $\tilde \Sigma_G$ that tracks the safe vector field $\kappa$. 
    \item Compute the control input $F = (F^S, F^G)$ to $\Sigma$ by:
    \begin{enumerate}
        \item Designing a PD tracking controller $F^S(v_s, v_g)$ for the $S$-subsystem $\Sigma_S$ that tracks $\Sc(\tilde F(v_g))$.
        \item Projecting $\tilde F$ onto $\F^G$, $F^G(v_s, v_g) = \Pc_{(s, g)}(\tilde F(v_g))$.
    \end{enumerate}
\end{enumerate}
This architecture encodes several layers of abstraction. First, we have the simplest model of the system, given by a kinematic reduced-order model $\kappa$ for $\Sigma$ on $G$. Next, we have a virtual mechanical model $\tilde \Sigma_G$ of the full system, for which we design a $\kappa$-tracking controller. Finally, we compute the input to the true system by tracking the virtual input. Altogether, this produces the desired $\kappa$-tracking architecture.

To construct a CBF from this architecture, we follow the approach of \cite{molnar2021model}. We will construct a certificate of the form
\begin{align*}
    h = h_0 - \sigma (\text{$\tilde \Sigma_G$ Tracking Cert.}) - \lambda (\text{$ \Sigma_S$ Tracking Cert.}),
\end{align*}
for $\sigma, \lambda > 0$. In the remainder of the section, we construct the controllers and tracking certificates for each layer of the architecture, resulting in a closed-form expression for $h$. 

\subsection{Constructing a Safe Hierarchical Controller}
 We begin by designing a controller for the virtual model $\tilde \Sigma_G$ that tracks the velocities of a safe vector field on $G$.
\begin{lem}\label{lem:virt-contr}
    Consider Setting \ref{setting:lay-arch}, and let $\kappa \in \X(G)$ satisfy $d(h_0)_g \kappa_g > -\alpha h_0(g) + \delta \norm{\grad h_0|_g}_G^2$ on $\C_0$, for $\alpha, \delta > 0$. Consider the following controller for the virtual model:
    \begin{align*}
        \tilde F(v_g) =  \big[\nabla_{v_g}^G\kappa + \grad V_G|_g -\tilde k_d\dot{\tilde e}\big]^\flat, \; \dot{\tilde e} = v_g - \kappa_g, 
    \end{align*}
    where $\tilde k_d > 0$ is a fixed constant. The Lie derivative of $v_g \mapsto \tfrac{1}{2}\lVert{\dot{\tilde e}}\rVert_G^2$ along trajectories of the closed-loop $\tilde \Sigma_G$ is $- \tilde k_d \lVert{\dot{\tilde e}}\rVert_G^2$.
\end{lem}
\begin{proof}
    The Lie derivative of $\tfrac{1}{2}\lVert{\dot{\tilde e}}\rVert_G^2$ at $v_g \in TG$ is
    \begin{align*}
        \tfrac{d}{dt}\tfrac{1}{2}\lVert{\dot{\tilde e}}\rVert_G^2 = \braket{ \dot{\tilde e}, - \grad V_G|_g + \tilde F(v_g)^\sharp - \nabla_{v_g}^G \kappa}.
    \end{align*}
    Substituting $\tilde F(v_g)$ yields $\frac{d}{dt}\tfrac{1}{2}\lVert{\dot{\tilde e}}\rVert_G^2 = -\tilde k_d \lVert{\dot{\tilde e}}\rVert_G^2$.
\end{proof}
The condition on $\kappa$ is a combination of a safety requirement (in $\alpha$) and a robustness requirement (in $\delta$). Outside of these requirements, $\kappa$ can be chosen freely to satisfy other performance specifications. For instance, $\kappa$ can be chosen as a smooth safety filter with an added robustness term \cite{de2026bundles}.

Next, using the shape map for $\Sigma$, we construct a reference trajectory for the next layer in the architecture from $\tilde F$ and $\Sc$. In order to effectively use the shape map $\Sc$, we shall require a set on which $\tilde F(v_g) \neq 0$ and $\tilde F(v_g) \in U$. To this end, fix $\sigma_0 > 0$, and consider the set:
\begin{align}
    N_{\sigma_0} = \{v_g \in TG : h_0(g) - \tfrac{\sigma_0}{2}\norm{\dot{\tilde e}}_G^2 \geq 0\}. 
\end{align}
Later, we shall build upon $N_{\sigma_0}$ to form our final safe set.
\begin{defn}\label{defn:ref}
    Consider Setting \ref{setting:lay-arch} and $\tilde F$ of Lemma \ref{lem:virt-contr}. Assume there is a $\sigma_0 > 0$ and a neighborhood $O$ of $N_{\sigma_0}$ for which $\tilde F|_O \neq 0$ and $\tilde F(O) \subset U$. Let $F^G(v_s, v_g) = \Pc_{(s, g)}(\tilde F(v_g))$.
    \begin{enumerate}
        \item The \textit{reference} induced by $\tilde F$ is the smooth map $r: O \to S$ defined $r(v_g) = \Sc(\tilde F(v_g))$.
        \item The \textit{reference velocity} induced by $\tilde F$ is the smooth map $\dot r: S \times O \to TS$ defined\footnote{Using local coordinates, one may show that $\dot r$ depends only on $(s, v_g)$.} $\dot r(s, v_g) = \frac{d}{dt}\big|_{t = 0} \Sc(\tilde F(\dot g(t)))$, for $(s(t), g(t))$ a trajectory of $\Sigma$ with $s(0) = s$, $\dot g(0) = v_g$, and $G$-controller $F^G$.
    \end{enumerate}
    We abbreviate $r(v_g), \dot r(s, v_g)$ as $r, \dot r$ when context allows.
\end{defn}
\begin{remark}
    When $\grad V_G \neq 0$, the assumption of a $\sigma_0 > 0$ for which $\tilde F \neq 0$ on a neighborhood $O$ of $N_{\sigma_0}$ is often straightforward to verify; when $\grad V_G \equiv 0$, this requires a more careful topological analysis. Due to space constraints, we do not pursue this analysis here.
\end{remark}
Taking advantage of the full actuation of the $S$-subsystem, we can track this reference using the geometric PD controller $F^S = F_{\text{PD}} + F_{\text{FF}}$ of \cite{bullo1999tracking}. $F_{\text{PD}}$ and $F_{\text{FF}}$ are defined as follows. For a trajectory $(s(t), g(t))$ of $\Sigma$ with  $G$-controller $F^G$, $r(t) = r(\dot g(t))$, $\dot e(t) = \dot s(t) - \T_{(s(t), r(t))}\dot r(t)$, and $k_d > 0$,
\begin{align*}
    F_{\text{PD}}  &= -d_1 \psi_{(s(t), r(t))} - k_d (\dot e(t))^\flat \\
    F_{\text{FF}}  &= \big[\nabla \T_{(s(t), r(t))}(\dot s(t), \dot r(t)) + \tfrac{d}{dt}\big|_{s = s(t)}\big(\T_{(s, r(t))}\dot r(t) \big)\big]^\flat.
\end{align*}
For each $(v_s, v_g)$, evaluating $F_{\text{PD}} + F_{\text{FF}}$ at $(\dot s(0), \dot g(0)) = (v_s, v_g)$ defines a control law $F^S(v_s, v_g)$ enabling $s$ to track $r$. As shown by \cite{bullo1999tracking}, the tracking performance of this controller can be certified through a careful analysis of two energy-like functions on $TS \times TS$:
\begin{align*}
    W(v_s, v_r) &= \psi(s, r) + \frac{1}{2}\norm{\dot e}_S^2, \; \dot e = v_s - \T_{(s, r)}v_r,\\
    W_\epsilon(v_s, v_r) &= W(v_s, v_r) + \epsilon \braket{d_1 \psi_{(s, r)}; \dot e}.
\end{align*}
In our setting, we will use $W_\epsilon$ to construct a candidate barrier function $h$, and $W$ to ensure that the $L$-regularity conditions and safety constraint are respected on the safe set of $h$.
\begin{lem}\label{lem:c0-bounds}
    Consider Setting \ref{setting:lay-arch}. Fix $\tilde F$ as in Lemma \ref{lem:virt-contr}, $\sigma_0$ and $O$ as in Definition \ref{defn:ref}, and $\epsilon_0 \in (0, \sqrt{2b_1})$. Then:
    \begin{enumerate}
        \item $W_{\epsilon_0} \geq c_0(\epsilon_0) W$ across $TS \times TS$, for $c_0(\epsilon) = 1 - \tfrac{\epsilon}{\sqrt{2b_1}}$.
        \item For $\lambda_0 \geq \frac{\sup_{\hat g \in \C_0} h_0(\hat g)}{L c_0(\epsilon_0)}$, the function $h^0: TS \times O \to \R$,
        \begin{align*}
            h^0(v_s, v_g) &= h_0(g) - \tfrac{\sigma_0}{2}\lVert{\dot{\tilde e}}\rVert_G^2 - \lambda_0 c_0(\epsilon_0) W(v_s, \dot r),
        \end{align*}
        is smooth and has compact $0$-superlevel set $\C^0$. Further, $h_0(g) \geq 0$ and $\psi(s, r) \leq L$ for all $(v_s, v_g) \in \C^0$. 
        \item The function $h: TS \times O \to \R$,
        \begin{align*}
            h(v_s, v_g) = h_0(g) - \tfrac{\sigma}{2}\lVert{\dot{\tilde e}}\rVert_G^2 - \lambda W_\epsilon(v_s, \dot r),
        \end{align*} 
        for $\sigma \geq \sigma_0, \; \lambda \geq \lambda_0$, and $\epsilon \in (0, \epsilon_0)$, is smooth and has $0$-superlevel set $\C \subseteq \C^0$.
    \end{enumerate}
\end{lem}
\begin{proof}
    First, we show (1). Using $b_1 \lVert{d_1 \psi_{(s, r)}}\rVert_S^2 \leq \psi(s, r)$, we bound $W_\epsilon \geq \psi(s, r) + \frac{1}{2}\norm{\dot e}_S^2 - \epsilon\sqrt{1/b_1} \sqrt{\psi(s, r)} \norm{\dot e}_S$. We can assemble this lower bound into a quadratic form with symmetric matrix $Q_{\epsilon_0}$; provided $\epsilon_0 \in (0, \sqrt{2b_1})$, this form will be positive definite. Similarly, $W$ can be written as a quadratic form in the same variables, with symmetric, positive definite $Q$. For the given $c_0$ function, $Q_{\epsilon_0} - c_0(\epsilon_0) Q$ is PSD, and item (1) follows.

    Next, we show (2). Since $\tilde F \neq 0$ on $O$, $h^0$ is well-defined and smooth. Fix $(v_s, v_g)$ for which $h^0(v_s, v_g) \geq 0$. Since $\frac{\sigma_0}{2}\lVert{\dot{\tilde e}}\rVert_G^2$ and $W$ are both non-negative, $h_0(g) \geq 0$. Further, $h_0(g) \geq \lambda_0 c_0(\epsilon_0) W(v_s, \dot r)$, implying $\sup_{\hat g \in \C_0} h_0(\hat g) \geq \lambda_0 c_0(\epsilon_0) W(v_s, \dot r)$. Given the choice of $\lambda_0$, this implies $L \geq W(v_s, \dot r)$, and therefore that $L \geq \psi(s, r)$. One may prove compactness of $\C^0$ using compactness of $\C_0$ and the properness of each map $\psi_r: s \mapsto\psi(s, r)$.

    Finally, we show item (3). Well-definedness and smoothness follow as in (2). Now, we show that $\C \subseteq \C^0$. Similar to (1), we find that $W_\epsilon \geq c_0(\epsilon) W$ for any $\epsilon \in (0, \epsilon_0)$. Since $c_0(\epsilon) W \geq c_0(\epsilon_0) W$ when $\epsilon < \epsilon_0$, it follows that $W_\epsilon \geq c_0(\epsilon_0) W$ for all $\epsilon \in (0, \epsilon_0)$. Thus, for $\sigma \geq \sigma_0, \lambda \geq \lambda_0$, and $\epsilon \in (0, \epsilon_0)$, $h(v_s, v_g) \leq h^0(v_s, v_g)$. Item (3) follows.
\end{proof}
Thus, for $\sigma \geq \sigma_0$, $\lambda \geq \lambda_0$, and $\epsilon \in (0, \epsilon_0)$, $\C \subseteq \C^0$. Thus, configuration safety, $L$-regularity conditions, and bounds derived on $\C^0$ all hold on $\C$ by inclusion. With this in mind, we bound the derivatives of the $G$ and $S$ components of $h$.
\begin{lem}\label{lem:lam-g-bound}
    Consider Setting \ref{setting:lay-arch}. Let $\kappa, \tilde F$ be defined as in Lemma \ref{lem:virt-contr}, $\C^0$ as in Lemma \ref{lem:c0-bounds}, and $\sigma \geq \sigma_0$. Under the controller $F^G(v_s, v_g) = \Pc_{(s, g)}(\tilde F(v_g))$, the Lie derivative of $\tilde h(v_g) = h_0(g) - \frac{\sigma}{2}\lVert{\dot{\tilde e}}\rVert_G^2$ at $(v_s, v_g) \in \C^0$ satisfies:
    \begin{align*}
        \dot{\tilde h} >& - \alpha h_0(g) +  \begin{bmatrix}
            \norm{\grad h_0|_g}_G\\
            \norm{\dot{\tilde e}}_G
        \end{bmatrix}^\top \Lambda_G \begin{bmatrix}
            \norm{\grad h_0|_g}_G\\
            \norm{\dot{\tilde e}}_G
        \end{bmatrix}\\
        &- \sigma B \norm{\dot{\tilde e}}_G \sqrt{C \psi(s,r)}\\
        \Lambda_G =& 
        \begin{bmatrix}
            \delta & -\frac{1}{2}\\
            -\frac{1}{2} & \sigma\tilde k_d(1 - \sqrt{CL})
        \end{bmatrix},
    \end{align*}
    for $\dot{\tilde e} = v_g - \kappa_g$ and $B \geq \sup_{v_g \in N_{\sigma_0}} \lVert{\nabla_{v_g}^G \kappa + \grad V_G|_g}\rVert_G$. 
\end{lem}
\begin{proof}
    Throughout, we write $\norm{\cdot}$ for $\norm{\cdot}_G$ and $\tilde F$ for $\tilde F(v_g)$. Adding and subtracting a $\tilde F$ term, we first compute:
    \begin{align}
        \dot{\tilde h}(v_s, v_g) &= d(h_0)_{g}v_g + \sigma \tilde k_d \norm{\dot{\tilde e}}^2  \label{eq:nom-deriv} \\
        &- \sigma\braket{\dot{\tilde e}, \big(\Pc_{(s, g)}(\tilde F) - \tilde F\big)^\sharp}. \label{eq:proj-error}
    \end{align}
    We start by bounding the first term in $\dot{\tilde h}$, \eqref{eq:nom-deriv}. Adding and subtracting $\kappa_g$ from $v_g$ in $d(h_0)v_g$, we bound:
    \begin{align}
        \eqref{eq:nom-deriv} &> -\alpha h_0(g) + \delta \norm{\grad h_0|_g}^2 \label{eq:nom-deriv-bound} \\ 
        &- \norm{\grad h_0|_g} \norm{\dot{\tilde e}} + \sigma \tilde k_d \norm{\dot{\tilde e}}^2. \nonumber 
    \end{align}
   We bound \eqref{eq:proj-error} using $L$-compatibility of $\Sc$ and $\psi$:
    \begin{align}
        \eqref{eq:proj-error} &\geq -\sigma \norm{\dot{\tilde e}} \lVert{\tilde F}\rVert  \lVert{(\Pc_{(s, g)}(\tilde F/\lVert{\tilde F}\rVert) - \tilde F/\lVert{\tilde F}\rVert)}\rVert \nonumber \\
        &\geq -\sigma \norm{\dot{\tilde e}} \lVert{\tilde F}\rVert \sqrt{C \psi(s, r)}. \label{eq:proj-error-bound}
    \end{align}    
    Combining \eqref{eq:nom-deriv-bound} and \eqref{eq:proj-error-bound} and applying $\lVert \tilde F \rVert \leq B + \tilde k_d \norm{\dot{\tilde e}}$,
    \begin{align}
        \dot{\tilde h} &> - \alpha h_0(g) + \delta \norm{\grad h_0|_g}^2 - \norm{\grad h_0|_g} \norm{\dot{\tilde e}}\\
        &+ \sigma \tilde k_d \norm{\dot{\tilde e}}^2
        - \sigma \tilde k_d \norm{\dot{\tilde e}}^2 \sqrt{C L} - \sigma B\norm{\dot{\tilde e}}\sqrt{C\psi(s, r)} , \nonumber
    \end{align}
    where we use that $\psi \leq L$ on $\C^0$. Collecting like terms and assembling a quadratic form with $\Lambda_G$, the result follows.
\end{proof}
Next, we bound the derivative of $W_\epsilon$, following \cite{bullo2005geometric}.
\begin{lem}\label{lem:lam-s-bound}
    Consider Setting \ref{setting:lay-arch}, and let $\tilde F$ be as in Lemma \ref{lem:virt-contr} and $\C^0$ as in Lemma \ref{lem:c0-bounds}. Let $F^S$ be a geometric PD controller for $\Sigma_S$ with $r, \dot r$ as in Definition \ref{defn:ref}. Then, there are $M_i \geq 0$ for which $\norm{\Hess_1 \psi_{(s, r)}}_S \leq M_1$, $\norm{\nabla \T_{(s, r)}}_S \leq M_2$, and $\norm{\dot r(s, v_g)}_S \leq M_3$ on $\C^0$. For $\epsilon \in (0, \epsilon_0)$ and $(v_s, v_g) \in \C^0$:
    \begin{align*}
        \dot W_\epsilon(v_s, \dot r(s, v_g)) &\leq - \begin{bmatrix}
            \norm{d_1 \psi_{(s, r)}}_{S}\\
            \norm{\dot e}_S
        \end{bmatrix}^\top \Lambda_S \begin{bmatrix}
            \norm{d_1 \psi_{(s, r)}}_{S}\\
            \norm{\dot e}_S
        \end{bmatrix}\\
        \Lambda_S &= 
        \begin{bmatrix}
            \epsilon & -\frac{\epsilon(k_d + M_2 M_3)}{2}\\
            -\frac{\epsilon(k_d + M_2 M_3)}{2} & k_d - \epsilon M_1
        \end{bmatrix}.
    \end{align*}
\end{lem}
\begin{proof}
    The existence of each $M_i$ follows from compactness of $\C^0$. To derive the bound, we recall from \cite[p. 652]{bullo2005geometric} that:
    \begin{align*}
        \dot W_\epsilon =& -k_d \norm{\dot e}_S^2 + \epsilon \Hess_1 \psi_{(s, r)}(\dot e, \dot e) - \epsilon\norm{d_1 \psi_{(s, r)}}_S^2\\
        &- \epsilon \braket{d_1 \psi_{(s, r)}; k_d \dot e + \nabla \T_{(s, r)}(\dot e, \dot r)}.
    \end{align*}
    Applying each $M_i$ bound and organizing like terms into a quadratic form, the result follows.
\end{proof}

Now, we put the individual bounds together to show that each controller can be designed to make $h$ a barrier function.
\begin{theo}\label{theo:h-cbf}
    Consider Setting \ref{setting:lay-arch}. Let $\tilde F$ be the controller of Lemma \ref{lem:virt-contr} and $\sigma_0 > 0$ be as in Definition \ref{defn:ref}, with:
    \begin{align*}
        \alpha > 0, \quad \delta > 0, \quad \tilde k_d > \tfrac{1}{2 \delta \sigma_0(1 - \sqrt{CL})}.
    \end{align*}
    Next, fix $\lambda_0, \epsilon_0$ as in Lemma \ref{lem:c0-bounds}. Let $F^S$ be a geometric PD controller with reference $r$ as in Definition \ref{defn:ref} and gain $k_d > 0$. Consider the function:
    \begin{align*}
        h(v_s, v_g) = h_0(g) - \tfrac{\sigma}{2}\lVert{\dot{\tilde e}}\rVert_G^2 - \lambda W_\epsilon(v_s, \dot r(s, v_g)),
    \end{align*}
    where $\epsilon, \sigma, \lambda$ satisfy:
    \begin{align*}
        \quad \epsilon &\in (0, \min\{\epsilon_0, \tfrac{2k_d}{2M_1 + (k_d + M_2 M_3)^2}\})\\
        \sigma &\geq \sigma_0, \quad \lambda \geq \max\{\lambda_0, \tfrac{\sigma B^2 C b_2}{\epsilon\tilde k_d(1 - \sqrt{CL})}\}.
    \end{align*}
    Define matrices $\hat \Lambda_G = \Lambda_G - \diag\big(0, \frac{1}{2}\sigma \tilde k_d(1 - \sqrt{CL})\big)$ and $\hat \Lambda_S = \Lambda_S - \diag\big(\frac{\epsilon}{2}, 0\big)$, and choose $C_G, C_S > 0$ for which:
    \begin{align*}
        \hat \Lambda_G \succeq C_G \begin{bmatrix}
            0 & 0\\
            0 & 1
        \end{bmatrix}, \quad \hat \Lambda_S \succeq C_S \begin{bmatrix}
            b_2 & \frac{\epsilon}{2}\\
            \frac{\epsilon}{2} & \frac{1}{2}
        \end{bmatrix}.
    \end{align*}
    If $\alpha \leq  \min\big\{\tfrac{2}{\sigma}C_G, C_S\big\}$, the Lie derivative of $h$ for the closed-loop system satisfies $\dot h(v_s, v_g) > - \alpha h(v_s, v_g)$ for all $(v_s, v_g)$ with $h(v_s, v_g) \geq 0$, and $h$ is a CBF for $\Sigma$ on $TS \times O$.
\end{theo}
\begin{proof}
    We will drop subscripts where clear; e.g. $\lVert{d_1 \psi_{(s, r)}}\rVert_S$ will be written $\norm{d_1\psi}$. Fix $(v_s, v_g) \in \C$. Given the choice of constants, $\C \subseteq \C^0$. Hence, all bounds derived on $\C^0$ hold and $g$ belongs to $\C_0$. Define the vectors $z_G = [\norm{\grad h_0}; \lVert{\dot{\tilde e}}\rVert]$ and $z_S = [\norm{d_1 \psi}; \norm{\dot e}]$. Combining Lemmas \ref{lem:lam-g-bound} and \ref{lem:lam-s-bound} with $\sqrt{C \psi} \leq \sqrt{Cb_2} \norm{d_1 \psi}$, we bound $\dot h(v_s, v_g)$ below by
    \begin{align*}
        > - \alpha h_0 + z_G^\top \Lambda_G z_G + \lambda z_S^\top \Lambda_S z_S - \sigma B\sqrt{Cb_2} \lVert{\dot{\tilde e}}\rVert\norm{d_1 \psi}.
    \end{align*}
    Next, we introduce the $\hat \Lambda_G$ and $\hat \Lambda_S$ matrices. We have:
    \begin{align}
        z_G^\top \Lambda_G z_G &= z_G^\top \hat \Lambda_G z_G + \tfrac{1}{2}\sigma \tilde k_d(1 - \sqrt{CL})\lVert{\dot{\tilde e}}\rVert^2  \label{eq:g-hat-eq}\\
        z_S^\top \Lambda_S z_S &= z_S^\top \hat \Lambda_S z_S + \tfrac{\epsilon}{2}\norm{d_1 \psi}^2. \label{eq:s-hat-eq}
    \end{align}
    Define a cross-term vector $z_C = [\norm{d_1 \psi}; \lVert \dot{\tilde e} \rVert]$. Substituting \eqref{eq:g-hat-eq} and \eqref{eq:s-hat-eq} into the $\dot h$ bound and collecting cross-terms,
    \begin{align}
        \dot h &> - \alpha h_0 + z_G^\top \hat \Lambda_G z_G + \lambda z_S^\top \hat \Lambda_S z_S + z_C^\top \hat \Lambda_C z_C, \label{eq:cross-term-bound}
    \end{align}
    where $z_C^\top \hat \Lambda_C z_C$ consists of all of the $z_C$-terms assembled into a quadratic form, with 
    \begin{align*}
        \hat \Lambda_C = \begin{bmatrix}
            \frac{1}{2} \lambda \epsilon & -\frac{1}{2}\sigma B \sqrt{Cb_2}\\
            -\frac{1}{2}\sigma B \sqrt{Cb_2} & \frac{1}{2}\sigma \tilde k_d(1 - \sqrt{CL})
        \end{bmatrix}.
    \end{align*}
    For the given choices of $\delta, \tilde k_d, \epsilon, \sigma, \lambda$, we have $\hat \Lambda_G, \hat \Lambda_S \succ 0$ and $\hat \Lambda_C \succeq 0$. Using $\hat \Lambda_C \succeq 0$ to drop $z_C^\top \hat \Lambda_C z_C$ and $C_G, C_S$ to lower bound $\hat \Lambda_G, \hat \Lambda_S$, we lower bound all of \eqref{eq:cross-term-bound} by:
    \begin{align}
        \eqref{eq:cross-term-bound} &\geq - \alpha h_0 + C_G \norm{\dot{\tilde e}}^2 + \lambda C_S z_S^\top \begin{bmatrix}
            b_2 & \frac{\epsilon}{2}\\
            \frac{\epsilon}{2} & \frac{1}{2}
        \end{bmatrix} z_S. \label{eq:zS-Weps-bound}
    \end{align}
    As the $z_S$ quadratic form upper bounds $W_\epsilon$, we conclude:
    \begin{align*}
        \eqref{eq:zS-Weps-bound} &\geq - \alpha h_0 + C_G \norm{\dot{\tilde e}}^2 + \lambda C_S W_\epsilon\\
        &\geq -\alpha h_0 + \alpha \tfrac{\sigma}{2}\norm{\dot{\tilde e}}^2 + \alpha \lambda W_\epsilon = -\alpha h(v_s, v_g)
    \end{align*}
    Thus, the closed-loop satisfies $\dot h(v_s, v_g) > -\alpha h(v_s, v_g)$ on $\C$. This implies that $\sup_{F \in \F_{(s, g)}} \dot h(v_s, v_g, F) > -\alpha h(v_s, v_g)$ on $\C$, and that $h$ is a CBF for $\Sigma$ defined on $TS \times O$.
\end{proof}
Thus, provided each layer of the proposed architecture is sufficiently robust and performant, the full control system remains safe \textit{and} is certified by a closed-form CBF. The bound on $\alpha$ naturally suggests that the aggression of the kinematic vector field $\kappa$ is limited by the performance of the tracking control laws.

Although $h$ \textit{is} a CBF like any other, and therefore can be implemented directly in a CBF-QP safety filter, it is often more practical to leverage a hierarchical perspective in the use of $h$. Instead of using $h$ as an end-to-end CBF, one may first design a safe velocity controller $\kappa$ and track it using the proposed hierarchical control design. Theorem \ref{theo:h-cbf} provides conditions under which $h$ is a certificate of safety for the \textit{full} architecture, even if it is not used in an end-to-end safety filter. This use of $h$ is reflective of the common design practice in which a safe controller is designed for a simple system and tracked by a performant low-level controller.

%% file: conclusion.tex
\section{Conclusion}
In this work, we performed a constructive, hierarchical study of safety-critical control for underactuated systems. We introduced shape maps, providing precise conditions for the existence of a correspondence between input directions in $G$ and shapes in $S$. Using these maps, we developed a layered control architecture enabling a hierarchical CBF synthesis procedure. Future work includes generalization to a broader class of underactuated systems.